%% file: Multiset-current-normal.tex
\documentclass[A4paper,USenglish,11pt]{article}
\def\fancyguy{0}
\include{includes}
\include{defines}
\begin{document}
\title{A Space-Efficient Dynamic Dictionary for Multisets with Constant Time Operations
\thanks{This research was supported by a
grant from the United States-Israel Binational Science Foundation
(BSF), Jerusalem, Israel, and the United States National Science
Foundation (NSF)}
}
\author{
Ioana O. Bercea\thanks{
Tel Aviv University, Tel Aviv, Israel.
Email:~\texttt{ioana@cs.umd.edu, guy@eng.tau.ac.il}.}
\and
Guy Even\footnotemark[2]
}
\date{}
\maketitle

\begin{abstract} 
We consider the dynamic dictionary problem for multisets. 
Given an upper bound $n$ on the total cardinality
of the multiset (i.e., including multiplicities) at any point in time,
the goal is to design a data structure that supports multiplicity
queries and allows insertions and deletions to the multiset (i.e., the dynamic setting).
The data structure must be space-efficient (the space is $1+o(1)$ times the
information-theoretic lower bound) and support all operations in
constant time with high probability. 

In this paper, we present the first dynamic dictionary for multisets
that achieves these performance guarantees.
This answers an open problem of Arbitman, Naor and Segev~\cite{arbitman2010backyard}.
The previously best-known construction of Pagh, Pagh and Rao~\cite{DBLP:conf/soda/PaghPR05}
supports membership in constant time, multiplicity queries in $O(\log n)$ time in the worst
case, and insertions and deletions in constant expected amortized time. 
The main technical component of our solution is a strategy for efficiently storing variable-length binary counters using weighted balls-into-bins experiments in which
balls have logarithmic weights. 

We also obtain a counting filter that approximates multiplicity queries with 
a one sided error, using the reduction of Carter~\etal~\cite{carter1978exact}.
Counting filters have received significant attention over the years due to their applicability in practice.
We present the first counting filter with constant time operations.

\end{abstract}

\section{Introduction}
We consider the dynamic dictionary problem for multisets. 
The special case in which every element of the universe can appear at most once is a fundamental
problem in data structures and has been well studied~\cite{arbitman2010backyard,DBLP:conf/soda/PaghPR05,raman2003succinct, demaine2006dictionariis}.
In the case of multisets, elements can have arbitrary multiplicities and
we are given an upper bound $n$ on the total cardinality
of the multiset (i.e., including multiplicities) at any point in time.
The goal is to design a data structure that supports multiplicity
queries and allows insertions and deletions to the multiset (i.e., the dynamic setting).

A related problem is that of supporting \emph{approximate} membership 
and multiplicity queries. The classic approximate setting allows one-sided errors
in the form of false positives: given an error parameter $\eps>0$, the probability
of returning a ``yes'' on an element not in the set must be upper bounded by $\eps$.
Such data structures are known as \emph{filters}. 
For multisets, the corresponding data structure is known as a \emph{counting} filter.
A counting filter returns a count that is at least 
the multiplicity of the element in the multiset and overcounts with probability
bounded by $\eps$. 
Counting filters have received significant attention over the years due to their
applicability in practice~\cite{fan2000summary, cohen2003spectral, bonomi2006improved}.
One of the main applications of dictionaries for multisets is precisely
in designing counting filters. Namely, Carter~\etal~\cite{carter1978exact} showed
that by hashing each element into a random fingerprint,
one can reduce a counting filter to a dictionary for multisets by storing
the fingerprints in the dictionary.

For the design of both dictionaries and filters, the performance measures of interest are the space the data structure takes
and the time it takes to perform the operations. 
For dictionaries, we would like to get close to the 
lower bound of $\log{u+n \choose n} = n\log(u/n) + \Theta(n)$ bits,
where $u$ is the size of the universe.\footnote{All logarithms are base $2$ unless otherwise
stated. $\ln x$ is used to denote the natural
logarithm.}\footnote{ This equality holds when $u$ is
significantly larger than $n$.}
In the case of filters, the lower bound is at least $n\log(1/\eps) + \Theta(n)$
bits~\cite{lovett2010lower}.
A data structure is \emph{space-efficient} if the total number of bits
it requires is within $(1+o(1))$ of the lower bound, where the $o(1)$ term converges
to zero as $n$ tends to infinity. 
The goal is to design data structures that are space-efficient with high probability.\footnote{By with high probability (whp), we mean with
probability at least $1-1/n^{\Omega(1)}$. The constant in the exponent
can be controlled by the designer and only affects the $o(1)$ term
in the space of the dictionary or the filter.}
We would like to support queries, insertions and deletions in constant time in the word RAM model.
The constant time guarantees should be in the worst case with high probability (see~\cite{broder2001using,kirsch2007using, arbitman2009amortized,arbitman2010backyard}
for a  discussion on the shortcomings of expected or amortized performance in practical scenarios).
We assume that each memory access can read/write a word of $w=\log n$ contiguous bits.

The current best known dynamic dictionary for multisets was designed by Pagh, Pagh, Rao~\cite{DBLP:conf/soda/PaghPR05}
based on the dictionary for sets of Raman and Rao~\cite{raman2003succinct}.
The dictionary is space-efficient and supports membership queries in constant time in the worst case.
Insertions and deletions take amortized expected constant time and multiplicity queries
take $O(\log n)$ in the worst case. 
In the case of sets, the state-of-the-art dynamic dictionary of Arbitman, Naor and Segev~\cite{arbitman2010backyard} achieves the ``best of both worlds'': it is space-efficient and supports all operations in constant time whp. 
Arbitman~\etal~\cite{arbitman2010backyard} leave it as an open problem whether a similar
result can be achieved for multisets. 

Recently, progress on this problem was achieved by Bercea and Even~\cite{BE20} who designed
a constant-time dynamic space-efficient dictionary for \emph{random} multisets. 
In a random multiset, each element is sampled independently and uniformly at random from the universe.
In this paper, we build upon their work and present the first space-efficient dynamic dictionary
for (arbitrary) multisets with constant time operations in the worst case with high probability,
resolving the question of Arbitman~\etal~\cite{arbitman2010backyard} in the positive.
We also obtain a counting filter with similar guarantees.

\subsection{Results}
In the following theorem, we assume that the size of the universe $\UU$ is
polynomial in $n$.\footnote{This is justified by mapping $\UU$ to
$[\poly(n)]$ using $2$-independent hash
functions~\cite{demaine2006dictionariis}.}
\emph{Overflow} refers to the event that the space
allocated in advance for the dictionary does not suffice.

\begin{theorem}[dynamic multiset dictionary]\label{thm:dict}
There exists a dynamic dictionary that maintains a multiset of
cardinality at most $n$ from the universe $\UU=\set{0,1}^{\log_2 u}$
with the following guarantees: (1)~For every polynomial in $n$
sequence of operations (multiplicity query, insertion,
deletion), the dictionary does not overflow whp. (2)~If the
dictionary does not overflow, then every operation can be completed
in constant time. (3)~The required space is
$(1+o(1))\cdot n\log (u/n)+O(n)$ bits.
\end{theorem}

Our dictionary construction considers a natural separation into the \emph{sparse}
and \emph{dense} case based on the size of the universe relative to $n$.
The sparse case, defined when $\log(u/n) = \omega(\log\log n)$, presents
a more straightforward challenge for dictionary design because the dictionary construction can afford to store  additional $\Theta(\log\log n)$ bits per element without sacrificing space-efficiency.
In this case, the dictionary for multisets is based on a simple observation. Namely,
elements with multiplicity at most $\log^3 n$ can be stored in a space-efficient dictionary for sets
by attaching to each element a fixed-length counter of $3\log\log n$ bits (see Section~\ref{sec:sparse}).

The majority of the paper is focused on designing a dictionary for
multisets in the \emph{dense} case, in which
$\log(u/n) = O(\log\log n)$.\footnote{This case is especially relevant
in the approximate membership setting in which we have $u/n=1/\eps$
due to the reduction of Carter~\etal~\cite{carter1978exact}. In this
setting, the dense case arises in applications in which $n$ is large
and the error probability $\eps$ is a constant (say $\eps=1\%$).}
Following~\cite{BE20}, we hash distinct elements into a first level
that consists of small space-efficient ``bin dictionaries'' of fixed
capacity.  
The first level only stores elements of multiplicity strictly smaller than
$\log^3 n$, just like in the dense case. 
However, we employ variable-length counters to encode multiplicities
and store them in a separate structure called a ``counter dictionary''. 
We allocate one counter dictionary per each bin dictionary.
The capacity of a counter dictionary is an upper bound on the total length of the counters 
it stores and is linear in the capacity of the associated bin dictionary.

Elements that do not fit in the first level are stored in a
secondary data structure called the spare.
The spare is small enough that it can allocate $\log n$ bit 
counters for the elements it stores.
To bound the number of elements that are stored in the ``spare'',
we cast the process of hashing counters into counter dictionaries as a
weighted balls-into-bins experiment in which balls have logarithmic weights
(see Sec.~\ref{sec:analysis}).

\medskip\noindent
As a corollary of Thm.~\ref{thm:dict}, we obtain a counting filter with the following
guarantees.\footnote{Note that we allow $\eps$ to be as small as $n/|\UU|$ (below this threshold, simply use a dictionary).}
\begin{corollary}[dynamic counting filter]\label{cor:filter}
There exists a dynamic counting filter for multisets of cardinality
at most $n$ from a universe $\UU=\set{0,1}^u$ such that the
following hold: (1)~For every polynomial in $n$ sequence of
operations (multiplicity query, insertion, deletion), the filter
does not overflow whp. (2)~If the filter does not overflow, then
every operation can be completed in constant time. (3)~The required
space is $(1+o(1))\cdot \log (1/\eps)\cdot n +O(n)$ bits. (4)~For
every count query, the probability of overcounting is bounded by $\eps$.
\end{corollary}

\subsection{Related Work}
The dictionary for multisets of Pagh~\etal~\cite{DBLP:conf/soda/PaghPR05} is space-efficient
and supports membership queries in constant time in the worst case. Insertions and deletions take
amortized expected constant time and multiplicity queries take $O(\log c)$ for a multiplicity of $c$.
Multiplicities are represented ``implicitly'' by a binary counter whose operations (query, increment, decrement)
are simulated as queries and updates to dictionaries on sets.\footnote{To be more exact,
for each bit of the counter, the construction in Pagh~\etal~\cite{DBLP:conf/soda/PaghPR05} allocates a dictionary on sets such that
the value of the bit can be retrieved by performing a lookup in the dictionary. Updating a bit of the
counter is done by inserting or deleting elements in the associated dictionary.}
Increments and decrements to the counter take $O(1)$ bit probes (and hence
$O(1)$ dictionary operations) but decoding the multiplicity
takes $O(\log n)$ time in the worst case. 
We are not aware of any other dictionary constructions for multisets.\footnote{Data structures for predecessor and successor queries such as~\cite{patrascu2014dynamic} can support multisets
but they do not meet the required performance guarantees
for the special case of (just) supporting multiplicity queries.}

Dynamic dictionaries for sets have been extensively studied~\cite{dietzfelbinger1990new, dalal2005two,
demaine2006dictionariis,raman2003succinct, fotakis2005space,panigrahy2005efficient,dietzfelbinger2007balanced,
arbitman2009amortized, arbitman2010backyard}. 
The dynamic dictionary for sets of Arbitman~\etal~\cite{arbitman2010backyard} is
space-efficient and supports operations in constant time whp.
Their construction cannot be generalized in a straightforward manner
to handle multisets. Specifically, their dictionary maintains
a spare of size $\Omega\parentheses{\frac{\log\log n}{(\log n )^{1/3}}\cdot n}$ elements and hence,
cannot store counters of length $\log n$ per element. 
In contrast, the spare in our construction is guaranteed to store at most $3n/(\log^3 n)$ elements whp.

In terms of counting filters, several constructions  do
not come with worst case guarantees for storing arbitrary
multisets~\cite{fan2000summary,
bonomi2006improved}.
The only previous counting filter with worst case guarantees we are
aware of is the Spectral Bloom filter of Cohen and
Matias~\cite{cohen2003spectral} (with over $450$ citations in Google
Scholar). The construction is a generalization of the Bloom filter and
hence requires $\Theta(\log(1/\eps))$ memory accesses per
operation. The space usage is similar to that of a Bloom filter and
depends on the sum of logs of multiplicities. Consequently, when the
multiset is a set, the leading constant is $1.44$, and hence Spectral
Bloom Filters are not space-efficient in general.

\subsection{Paper Organization}
Preliminaries are in Sec.~\ref{sec:prelim}. The construction
for the sparse case can be found in Sec.~\ref{sec:sparse}
and the one for the dense case is described and analyzed in Sec.~\ref{sec:dense}.
Section~\ref{sec:hash} describes how our analysis works without the assumption
of access to truly random hash functions. Corollary~\ref{cor:filter} is
proved in Sec.~\ref{sec:filter}. Appendix~\ref{app:bin} reviews standard
implementation techniques.

\section{Preliminaries}\label{sec:prelim}
\subsection{Notation and Definitions}
For $k>0$, let $[k]$ denote the
set $\set{0,\ldots,\ceil{k}-1}$.  For a string $a \in \set{0,1}^*$,
let $\size{a}$ denote the length of $a$ in bits.  We often abuse
notation, and regard elements in $[k]$ as binary strings of length
${\log k}$.  Let $\UU=\set{0,1}^{\log u}$ denote the universe of all
possible elements.

\begin{definition}[multiset]
A \emph{multiset} $\MM$ over $\UU$ is a function $\MM:\UU\rightarrow
\NN$. We refer to $\MM(x)$ as the \emph{multiplicity} of $x$.

\medskip\noindent The \emph{cardinality} of a multiset $\MM$ is
denoted by $\size{\MM}$ and defined by
$\size{\MM}\triangleq \sum_{x\in \UU} \MM(x)$.  The \emph{support} of
the multiset is denoted by $\sigma(\MM)$ and is defined by
$\sigma(\MM)\triangleq \set{x \mid \MM(x)> 0}$.
\end{definition}

\medskip\noindent \textbf{Operations over Dynamic Multisets.}  We
consider the following operations: $\ins(x)$, $\del(x)$, and $\cquery(x)$.  Let
$\MM_t$ denote the multiset after $t$ operations. A \emph{dynamic
multiset} $\set{\MM_t}_t$ is specified by a sequence $\set{\op_t}_{t\geq 1}$
of as follows.\footnote{We require that $\op_t=\del(x_t)$ only if
$\MM_{t-1}(x_t)>0$, i.e. if $x$ is not in the multiset, then a
delete operation does not make its multiplicity negative.}
\begin{align*}
\MM_t(x)&\triangleq
\begin{cases}
0 & \text{if $t=0$}\\
\MM_{t-1}(x)+1 &\text{if $\op_t=\ins(x)$}\\
\MM_{t-1}(x)-1 &\text{if $\op_t=\del(x)$}\\
\MM_{t-1}(x)& \text{otherwise.}
\end{cases}
\end{align*}
\medskip\noindent We say that a dynamic multiset $\set{\MM_t}_t$ has
cardinality at most $n$ if $\size{\MM_t}\leq n$, for every $t$.

\medskip\noindent \textbf{Dynamic Dictionary for Multisets.}  A
\emph{dynamic dictionary for multisets} maintains a dynamic multiset
$\set{\MM_t}_t$.  The response to $\cquery(x)$ is simply $\MM_t(x)$.

\medskip\noindent \textbf{Dynamic Counting Filter.}  A \emph{dynamic
counting filter} maintains a dynamic multiset $\set{\MM_t}_t$ and is
parameterized by an error parameter $\eps\in (0,1)$.  Let $\out_t$
denote the response to a $\cquery(x_t)$ at time $t$.  We require that
the output $\out_t$ satisfy the following conditions:
\begin{align}
\out_t \geq \MM_t(x_t)\\
\Pr{\out_t > \MM_t(x_t)} & \leq \eps\;.
\end{align}
Namely, $\out_t$ is an approximation of $\MM_t(x_t)$ with a
one-sided error.  

\begin{definition}[overcounting]
Let $\Err_t$ denote the event that $\op_t=\cquery(x_t)$, and
$\out_t>\MM_t(x_t)$.
\end{definition}
Note that overcounting generalizes false positive events in filters over
sets. Indeed, a false positive event occurs if $\MM_t(x_t)=0$ and
$\out_t>0$.\footnote{ The
probability space is induced only by the random choices (i.e.,
choice of hash functions) that the filter makes. Note also that if
$\op_t=\op_{t'}=\cquery(x)$,
then the events $\Err_{t}$ and $\Err_{t'}$ need not be independent.
}

\subsection{The Model}
\medskip\noindent \textbf{Memory Access Model.} We assume that the
data structures are implemented in the RAM model in which the basic
unit of one memory access is a word. Let $w$ denote the memory word
length in bits. We assume that $w=\Theta(\log n)$.  See
Appendix~\ref{app:bin} for a discussion on how the computations we perform
over words are implemented in constant time.

\medskip\noindent \textbf{Success Probability.}  We prove that
overflow occurs with probability at most $1/\poly(n)$ and that one can
control the degree of the polynomial (the degree of the polynomial
only affects the $o(1)$ term in the size bound).  The probability of
an overflow depends only on the random choices that the dictionary makes.

\medskip\noindent \textbf{Hash Functions.}
Our dictionary uses the succinct hash functions of Arbitman~\etal~\cite{arbitman2010backyard}
which have a small representation and can be evaluated in constant time.
For simplicity, we first analyze the data structure assuming fully random
hash functions (Sec.~\ref{sec:analysis}).  In Sec.~\ref{sec:hash}, we prove
that the same arguments hold when we use succinct hash functions.
The filter reduction additionally employs pairwise independent hash functions.

\section{Dictionary for Multisets via Dictionary+Retrieval (Sparse Case)}\label{sec:sparse}

In this section, we show how to design a multiset dictionary using any
dictionary on sets that supports attaching satellite data of
$O(\log n)$ bits per element.  Such a dictionary with satellite data
supports the operations: query, insert, delete, retrieve, and update.
A retrieve operation for $x$ returns the satellite data of $x$.  An
update operation for $x$ with new satellite data $d$ stores $d$ as the
new satellite data of $x$.  The reduction incurs a penalty of
$\Theta(\log\log n)$ extra bits per element. Hence, a space-efficient
multiset dictionary is obtained from a space-efficient dictionary only
if $\log(u/n)=\omega(\log\log n)$.

Let $\Dict(n,r)$ denote a dynamic dictionary for sets of cardinality
at most $n$, where $r$ bits of satellite data are attached to each
element.
Let $\MSDict(n)$ denote a dynamic dictionary for multisets of
cardinality at most $n$.

\medskip\noindent
The reduction is summarized in the following observation.
\begin{observation}\label{obs:sparse}
One can implement $\MSDict(n)$ using two dynamic dictionaries:
$D_1=\Dict(n,3\log\log n)$ and $D_2=\Dict(n/(\log^3 n), \log n)$.
Each operation over $\MSDict$ can be performed using a constant number
of operations over $D_1$ and $D_2$.
\end{observation}
\begin{proofsketch}
An element is \emph{light} if its multiplicity is at most $\log^2 n$,
otherwise it is \emph{heavy}.  Dictionary $D_1$ is used for storing
the light elements, whereas dictionary $D_2$ is used for storing the
heavy elements. The satellite data in both dictionaries is a binary
counter of the multiplicity. 
\end{proofsketch}

\begin{claim2}
If $\log(u/n)=\omega(\log \log n)$, then there exists a dynamic
multiset dictionary that is space-efficient and supports operations
in constant time in the worst case whp.
\end{claim2}
\begin{proof}
A space-efficient implementation of $\Dict(n,r)$ (for $r=O(\log n)$)
with constant time per operation can be obtained from the dictionary
of Arbitman~\etal~\cite{arbitman2010backyard} (see
also~\cite{BE20}). The space of such a dictionary is
$(1+o(1))\cdot (\log (u/n)+r) \cdot n + O(n)$ bits.  Instantiating
this space for $D_1$ and $D_2$ from Observation~\ref{obs:sparse}
yields a multiset dictionary $\MSDict(n)$ with space:
$(1+o(1))\cdot ((\log(u/n)+3\log \log n)\cdot n + O(n)$.  In the
sparse case $\log(u/n)=\omega(\log \log n)$, and hence the obtained
$\MSDict(n)$ is space efficient.
\end{proof}

\medskip\noindent
This completes the proof of Theorem 1 for the sparse case.

\medskip\noindent \textbf{Remark.} An alternative solution stores the
multiplicities in an array separately from a dictionary that stores
the support of the multiset.  Let $s$ denote the cardinality of the
support of the multiset.  Let $h:\UU\rightarrow [s+o(s)]$ be a dynamic
perfect hashing that requires $\Theta(s\log\log s)$ bits and supports
operations in constant time (such as the one
in~\cite{demaine2006dictionariis}).  Store the (variable-length)
binary counter for $x$ at index $h(x)$ in the array.  The array can be
implemented in space that is linear in the total length of the
counters and supports query and update operations in constant
time~\cite{blandford2008compact}.

\section{Dictionary for Multisets (Dense Case)}\label{sec:dense}
In this section,  we  prove Theorem~\ref{thm:dict} for the case
in which $\log(u/n) = O(\log\log n)$, which we call the \emph{dense} case.
We refer to this dictionary construction as the \emph{MS-Dictionary} (Multiset Dictionary) in the dense case.

The MS-Dictionary construction follows the same general structure as
in~\cite{arbitman2010backyard,demaine2006dictionariis,BE20}.
Specifically, it consists of two levels of dictionaries.  The first
level is designed to store the majority of the elements
(Sec.~\ref{sec:firstlevel}).  An element is stored in the first level
provided that its multiplicity is at most $\maxcard$ and
there is enough capacity.  Otherwise, the element is stored in the
second level, which is called the \emph{spare} (Sec.~\ref{sec:spare}).

The first level of the MS-Dictionary consists of $m$ \emph{bin
dictionaries} $\set{\BD_i}_{i\in[m]}$ together with $m$
\emph{counter dictionaries} $\set{\CD_i}_{i\in [m]}$.  Each bin
dictionary can store at most $n_B = (1+\delta)B$ distinct elements,
where $\delta=o(1)$ and $B\triangleq n/m$ denotes the mean occupancy
of each bin dictionary.  We say that a bin dictionary is \emph{full}
if it stores $n_B$ elements in it.

Each counter dictionary stores variable-length binary counters.  Each
counter represents the multiplicity of an element in the associated
bin dictionary.  Each counter dictionary can store counters whose
total length in bits is at most $6B$. We say that a counter dictionary
is \emph{full} if the total length of the counters stored in it is
$6B$ bits.

Elements with high multiplicity or whose $\BD$ or $\CD$ are full are
stored in the spare, as formulated in the following invariant:
\begin{invariant}\label{inv:first}
An element $x$ such that $\MM_t(x)>0$ is stored in the spare at time
$t$ if: (1)~$\MM_t(x)\geq\maxcard$, (2)~the bin dictionary
corresponding to $x$ is full, or (3)~the counter dictionary
corresponding to $x$ is full.
\end{invariant}

We denote the upper bound on the cardinality of the support of the
multiset stored in the spare by $n_S$ (the
value of $n_s$ is specified later).  We say that the spare \emph{overflows}
when more than $n_S$ elements are stored in it.

\subsection{Hash Functions}
We employ a permutation $\pi:\UU \rightarrow \UU$ .  We define
$\hb:\UU\rightarrow [m]$ to be the leftmost $\log m$ bits of the
binary representation of $\pi(x)$ and by $\hr:\UU\rightarrow [u/m]$ to
be the remaining $\log(u/m)$ bits of $x$.  An element $x$ is hashed to
the bin dictionary of index $\hb(x)$.  Hence storing $x$ in the first
level of the dictionary amounts to storing $\hr(x)$ in $\BD_i$, where
$i=\hb(x)$, and storing $\MM_t(x)$ in $\CD_i$. (This reduction in the universe size is often called
``quotienting''~\cite{Knuth, pagh2001low,
DBLP:conf/soda/PaghPR05,demaine2006dictionariis}).

The overflow analysis in Sec.~\ref{sec:analysis} assume truly random permutations.
In Sec.~\ref{sec:hash}, we discuss how one can replace this assumption
with the succinct hash functions of Arbitman~\etal~\cite{arbitman2010backyard}.

\subsection{The First Level of the Dictionary}\label{sec:firstlevel}
We follow the same parametrization as in~\cite{BE20}. Namely, we set
the average occupancy of a bin dictionary to be
$B \triangleq (\log n)/\log(u/n)$ and  set $\delta \triangleq \Theta(\frac{\log\log n}{\sqrt{B}})$.

\medskip\noindent
\textbf{Bin Dictionaries.}
Each bin dictionary ($\BD$) is a deterministic dictionary for sets of
cardinality at most $n_B$ that supports queries, insertions and
deletions. The implementation of a bin dictionary using global lookup
tables~\cite{arbitman2010backyard} or Elias-Fano encoding~\cite{BE20}
is briefly reviewed in Appendix~\ref{app:bin}.  We remark that each
$\BD$ is space-efficient, meaning it requires
$n_B \cdot \log(u/n) + O(n_B)$ bits.  Moreover, each $\BD$ fits in a
constant number of words and performs queries, insertions and
deletions in constant time.

\medskip\noindent
\textbf{Counter Dictionaries.}
Each counter dictionary $\CD_i$ stores a vector of multiplicities of
the elements stored in the corresponding bin dictionary $\BD_i$.  The
order of the multiplicities stored in $\CD_i$ is the same order in
which the corresponding elements are stored in $\BD_i$.
Multiplicities in $\CD_i$ are stored by variable-length counters. We
employ a trivial $2$-bit alphabet to encode $0,1$ and
``end-of-counter'' symbols for encoding the multiplicities. Hence, the
length of a counter $c$ is $\ceil{\log_2 c}$ bits and its encoding
$2(1+\ceil{\log_2 c})$ bits long. The contents of $\CD_i$ is simply a
concatenation of the encoding of the counters.  We allocate
$2(6B+n_B)$ bits per $\CD$.\footnote{Note, however, that we define a
$\CD$ to be full if the sum of counter lengths is $6B$ (even if we
did not use all its space).  The justification for this definition
is to simplify the analysis.}

The $\CD$ supports the operations of multiplicity query, increment and
decrement.  These operations are carried out naturally in
constant time because each $\CD_i$ fits in a word.  We note that an
increment may cause the $\CD$ to be full, in which case $x$ is deleted
from the bin dictionary and is inserted to the spare together with its
updated counter.  Similarly, a decrement may zero the counter, in
which case $x$ is deleted from the bin dictionary (and hence its
multiplicity is also deleted from the counter dictionary).

\subsection{The Spare}\label{sec:spare}
The spare is a high performance space-inefficient dictionary for multisets.
It stores at most $n_S = O(n/\log^3 n)$ distinct elements. 
Each element stored in the spare can have a multiplicity as high as $n$.  
It supports all the operations of the dictionary in constant time. 
In addition, the spare also moves elements back to the first level if their insertion no longer violates Invariant~\ref{inv:first}. 

We propose to implement the spare using the dynamic dictionary of
Arbitman~\etal~\cite{arbitman2009amortized} in which we append
$\log n$-bit counters to each element. We briefly review the
construction here.  The dictionary is a de-amortized construction of
the cuckoo hash table of Pagh and Rodler~\cite{pagh2001cuckoo}.
Namely, each element is assigned two locations in an array. If upon
insertion, both locations are occupied, then space for the new element
is made by ``relocating'' an element occupying one of the two
locations.  Long chains of relocations are ``postponed'' by employing a
queue of pending insertions.  Thus, each operation is guaranteed to
perform in constant time in the worst case.  The space that the
dictionary occupies is $O(n_S\log(u/n))+O(n)$.  The counters increase
the space of the spare by $O(n_S\log n) = o(n)$ bits.

The construction in~\cite{arbitman2009amortized} is used as a spare in
the space-efficient dynamic filter in~\cite{arbitman2010backyard}.  We
use it a similar manner to maintain Invariant~\ref{inv:first} in a
``lazy'' fashion. Namely, if an element $x$ residing in the spare is
no longer in violation of Invariant~\ref{inv:first} (for instance, due
to a deletion in the bin dictionary), we do not immediately move $x$
from the spare back to its bin dictionary.  Instead, we ``delay'' such
an operation until $x$ is examined during a chain of relocations.
Specifically, during an insertion to the spare, for each evicted
element, one checks if this element is still in violation of
Invariant~\ref{inv:first}.  If it is not, then it is deleted from the
spare and inserted into the first level.  This increases the time it
takes to perform an insertion to the spare only by a constant.
Moreover, it does not affect the overflow probability of the spare.

\subsection{Overflow Analysis }\label{sec:analysis}
The event of an overflow occurs if more than $n_S$ distinct elements
are stored in the spare. In this section, we prove that overflow does
not occur whp with respect to perfectly random hash
functions.  In Sec.~\ref{sec:hash}, we discuss how this analysis can
be modified when we employ succinct hash functions.

The analysis proceeds in two stages. First, we consider the
incremental setting (in which elements of the multiset are inserted
one-by-one and there are no deletions). We prove that overflow does
not occur whp if $n_S = 3n/\log^3(n)$. The proof for the dynamic
setting (deletions and insertions) is based on
Invariant~\ref{inv:first}. Namely, Invariant~\ref{inv:first} reduces
the dynamic setting to an incremental setting. Formally, the
probability of overflow at time $t$ (after a sequence of deletions and
insertions) equals the probability of an overflow had the elements of
$\MM_t$ been inserted one-by-one (no deletions).  Hence, overflow does
not occur whp over a polynomial number of operations in the dynamic
setting by applying a union bound.

Recall that each component of the first level of the dictionary has
capacity parameters: each bin dictionary has an upper bound of
$n_B = (1+\delta)B$ on the number of distinct elements it stores and
each counter dictionary has an upper bound of $6B$ on the total length
of the counters it stores.  Additionally, the first level only stores
elements whose multiplicity is strictly smaller than $\maxcard$.
According to Invariant~\ref{inv:first}, if the insertion of some
element $x$ exceeds these bounds, then $x$ is moved to the spare.

We bound the number of elements that go to the spare due to failing
one of the conditions of Invariant~\ref{inv:first} separately. The
number of elements whose multiplicity is at least $\maxcard$ is at
most $n/\maxcard$.  The number of distinct elements that are stored in
the spare because their bin dictionary is full is at most $n/\log^3 n$
whp. The proof of this bound can be derived by modifying the proof of Claim~\ref{claim:cdovf} (see also ~\cite{arbitman2010backyard}).
We focus on the number of
distinct elements whose counter dictionary is full.

\begin{claim2}\label{claim:cdovf}
The number of distinct elements whose corresponding $\CD$ is full is at most $n/\log^3 n$ whp.
\end{claim2}
\begin{proof}
Recall that there are $m = n/B$ counter dictionaries and that each
$\CD$ stores the multiplicities of at most $n_B = (1+\delta)B$
distinct elements of multiplicity strictly smaller than
$\maxcard$. In a full $\CD$, the sum of the counter lengths reaches
$6B$. We start by bounding the probability that the total length of
the counters in a $\CD$ is at least $6B$.

Formally, consider a multiset $\MM$ of cardinality $n$ consisting of
$s$ distinct elements $\set{x_i}_{i\in [s]}$ with multiplicities
$\set{f_i}_{i\in [s]}$ (note that $\sum_{i\in [s]} f_i =n$).  The
length of the counter for multiplicity $f_i$ is
$w_i\triangleq \ceil{\log(f_i+1)}$ (we refer to this quantity as
\emph{weight}).  For $\beta\in [m]$, let $\MM^\beta$ denote the
sub-multiset of $\MM$ consisting of the elements $x_i$ such that
$h^b(x_i)=\beta$.  Let $C_\beta$ denote the event that the weight of
$\MM^\beta$ is at least $6B$, namely
$\sum_{x_i\in\MM^\beta} w_i \geq 6B$. We begin by bounding the
probability of event $C_\beta$ occurring.

For $i\in [s]$, define the random variable $X_i\in \set{0,w_i}$,
where $X_i=w_i$ if $\hb(x_i)=\beta$ and $0$ otherwise.  Since the
values $\set{(\hb(x_i),\hq(x_i))}_i$ were sampled at random
without replacement (i.e., obtained from a random permutation),
the random variables $\set{X_i}_i$ are negatively associated.  Let
$\mu\triangleq \frac{1}{m} \cdot \sum_{i\in [s]} w_i$ denote the
expected weight per $\CD$.  Clearly, $\mu\leq \frac{n}{m} = B$.  We
now scale the RVs so that they are in the range $[0,1]$.  Since
the multiplicities of elements in the first level is strictly
smaller than $\maxcard$, we have that $w_i \leq \log \maxcard$ (we
omit the ceiling to improve readability). We then define
$\tilde{X_i}\triangleq X_i/\log \maxcard$ and
$\tilde{\mu}=\mu/\log \maxcard\leq B/\log \maxcard$. Then, by
Chernoff's bound:

\begin{align*}
\Pr{C_\beta} &= \Pr{\sum_{i\in [s]} X_i \geq 6B} \\
       &=\Pr{\sum_{i\in [s]} \tilde{X}_i \geq \frac{6B}{\log\maxcard}} \\
       &\leq 2^{-\frac{6B}{\log\maxcard}}\\
       &=1/(\log n)^{\omega(1)} \;.
\end{align*}

Let $I(C_\beta)$ denote the indicator variable for event $C_\beta$. Then $\expectation{}{\sum_\beta I(C_\beta)} \leq n/(\log n)^{\omega(1)}$.
Moreover, the RVs
$\set{I(C_\beta)}_\beta$ are negatively associated (more weight in bin $b$
implies less weight in bin $b'$). By Chernoff's bound:
\begin{align*}
\Pr{\sum_{b} I(C_\beta) \geq \frac{n}{\log^5 n}} \leq O(2^{- n/(\log^5 n)})\;.  
\end{align*}
Whp, a bin is assigned at most $\log^2 n$ elements.  We conclude that
the number of elements that are stored in the spare due to events
$\bigcup_{b} C_\beta$ is at most $n/(\log^3 n)$ whp.
\end{proof}

\subsection{Space Analysis}
Each bin dictionary takes $n_B\log(u/n) + \Theta(n_B)$ bits, where $n_B = (1+\delta)B$, $B = n/m$ and $\delta=o(1)$.
Each $\CD$ occupies $\Theta(B)$ bits. 
Therefore, the first level of the MS-Dictionary takes $(1+\delta)n\log(u/n) + \Theta(n)$ bits. The spare takes $O(n_S\log(u/n))=o(n)$ bits, since
$n_S = \Theta(n/\log^3 n)$. Therefore, the space the whole dictionary takes is $(1+o(1))\cdot \log(u/n)+\Theta(n)$ bits. This completes the proof of Theorem~\ref{thm:dict} for the dense case.

\section{Succinct Hash Functions}\label{sec:hash}
In this section, we discuss how to replace the assumption of truly random permutations 
with succinct hash functions (i.e., representation requires
$o(n)$ bits) that have constant evaluation time in the RAM model. 

We follow the construction in~\cite{arbitman2010backyard}, which we
describe as follows.  Partition the universe into $M = n^{9/10}$ parts
using a one-round Feistel permutation (described below) such that the number of elements
in each part is at most $n^{1/10}+n^{3/40}$ whp.
The permutation
uses highly independent hash
functions~\cite{siegel2004universal,dietzfelbinger2009applications}.
Apply the dictionary construction separately in each part with an upper
bound of $n^{9/10}+n^{3/40}$ on the cardinality of the set.  Within
each part, the dictionary employs a $k$-wise $\delta$-dependent
permutation.  A collection $\Pi$ of permutations
$\pi:\UU\rightarrow \UU$ is $k$-wise $\delta$-dependent if for any
distinct elements $x_1,\ldots, x_k \in \UU$, the distribution on
$(\pi(x_1),\ldots, \pi(x_k))$ induced by sampling $\pi \in \Pi$ is
$\delta$-close in statistical distance to the distribution induced by
a truly random permutation.
Arbitman~\etal~\cite{arbitman2010backyard} show how one can obtain
succinct $k$-wise $\delta$-dependent permutations that can be evaluated
in constant time by combining the constructions
in~\cite{naor1999construction,kaplan2009derandomized}.  Setting
$k=n^{1/10}+n^{3/40}$ and $\delta = 1/n^{\Theta(1)}$ ensures that the
bound on the size of the spare holds whp in each part and hence, by
union bound, in all parts simultaneously.

To complete the proof, we need to prove that the partitioning is
``balanced'' whp also with respect to multisets.
(Recall, that the cardinality of a multiset equals the sum of
multiplicities of the elements in the support of the multiset.)
Formally, we prove that the pseudo random partition induces
in each part a multiset of cardinality at most
$n^{1/10}+n^{3/40}\log^{3/2}n$ whp. As ``heavy'' elements of
multiplicity at least $\maxcard$ are stored in the spare, we may
assume that multiplicities are less that $\maxcard$. 

We first describe how the partitioning is achieved in~\cite{arbitman2010backyard}. The binary representation of $x$
is partitioned into the leftmost $\log M$ bits, denoted by $x_L$ and the remaining bits,
denoted by $x_R$. A $k'$-wise independent hash function $f:\set{0,1}^{\log(u/M)}\rightarrow\set{0,1}^{\log M}$
is then sampled, with $k'= \floor{n^{1/20}/(e^{1/3})}$. The permutation $\pi$ is defined as $\pi(x) = (x_L\oplus f(x_R),x_R)$. 

Note that this
induces a view of the universe as a two-dimensional table with $u/M$
rows (corresponding to each $x_R$ value) and $M$ columns
(corresponding to each $x_L\oplus f(x_R)$ value).  
Indeed, each cell of the table has at most one element
(i.e., if
$x=(x_L, x_R)$ and $y=(y_L,y_R)$ satisfy
$x_L\oplus f(x_R) = y_L \oplus f(y_R)$ and $x_R = y_R$, then $x=y$).
We define a \emph{part} of the input multiset as consisting of all the elements of the input multiset 
that belong to the same column. The index of the part that $x$ is assigned to is $x_L\oplus f(x_R)$.
The corresponding part stores $x_R$.

The following observation follows from ~\cite[Claim 5.4]{arbitman2010backyard}
and the fact that the maximum multiplicity of each element is strictly less than $\maxcard$.

\begin{observation}
The cardinality of every part of the multiset is at most
$n^{1/10}+n^{3/40}\log^{3/2}n$ whp.
\end{observation}

\begin{proof} 
Fix a part $j \in [M]$ and for each $i\in[u/M]$, let $\MM_i$ denote the multiset of all
elements $x$  with the $x_R$ value equal to $i$
(i.e., the  multisets $\MM_i$ consist of all the elements in row $i$).
Each multiset $\MM_i$ contributes at most one distinct element to the multiset of part $j$.
Define $X_i \in [\maxcard]$ to be the random variable that denotes the multiplicity of the element from $\MM_i$ that is mapped to part $j$.
Then $\expectation{}{X_i} =\frac{1}{M}\ \sum_{x\in \UU} \MM_i(x)$. 
Now define $X=\sum_{i\in[u/m]} X_i$ to be the random variable that denotes the cardinality of the multiset
that is mapped into part $j$. By linearity of expectation, $\expectation{}{X} = n/M= n^{1/10}$.
The random variables $\set{X_i}_i$ are $k'$-wise independent, since each variable $X_i$
is determined by a different row in the table (and hence, each $\set{X_i}_i$ depends on a different $x_R$ value). We scale the RVs $\set{X_i}_i$ by $\maxcard$ and then apply Chernoff's bound for $k'$-wise independent RVs~\cite{schmidt1995chernoff} and obtain:
\begin{align*}
\Pr{\frac{X}{\maxcard} \geq \parentheses{1+\frac{\log^{3/2}n}{n^{1/40}}}\cdot \frac{n^{9/10}}{\maxcard}} &\leq \exp(-\floor{k'/2}) = \exp(-\Omega(n^{1/20}))\;.
\end{align*}
The claim follows.
\end{proof}

\section{The Counting Filter}\label{sec:filter}
To obtain a counting filter from our dictionary for multisets, use a
pairwise independent hash function $h:\UU\rightarrow [n/\eps]$ to map
an element $x$ to a fingerprint $h(x)$\cite{carter1978exact}.  Let
$\MM_h$ denote the multiset over $[n/\eps]$ induced by a multiset
$\MM$ over $\UU$ defined by
$\MM_h(y) \triangleq \sum_{x:h(x)=y} \MM(x)$.  A multiset dictionary
for $\MM_h$ constitutes a counting filter in which the
probability of an overcount is at most $\eps$.  The counting filter
is requires $(1+o(1))\log(1/\eps)n + O(n)$ bits and performs all
operations in constant time. This completes the proof of
Corollary~\ref{cor:filter}.

\bibliography{main}

\appendix

\section{Implementation of the First Level of the Dictionary}\label{app:bin}
In this section, we discuss two implementations of the first
level of the dictionary that meet the specifications from Sec.~\ref{sec:firstlevel}.
Namely, that of using global lookup tables like it was suggested in~\cite{arbitman2010backyard} or an Elias-Fano encoding~\cite{elias1974efficient}. 
We briefly review them here (for details, see~\cite{BE20}).

\subsection{Global Lookup Tables}
In this implementation, all bin and counter dictionaries employ
common global lookup tables. Hence, it is sufficient to show that
the size of the tables is $o(n)$. Each bin dictionary stores at most $n_B=(1+\delta)B$ distinct elements
from a universe $\UU'$ of size $u'=u/(n/B)$. Therefore, the
total number of states of a bin dictionary is $s\triangleq\binom{u'}{n_B}$.
Each operation on the bin dictionary ( query, insert, delete) 
is implemented as a function from $s\times u'$ to $s$., Namely, given  the current
state of the dictionary and an element $x\in \UU'$, each function
returns an updated state (in the case of $\ins(x)$ and $\del(x)$) or a bit (in the case
of a membership query). 
The global lookup tables explicitly represent these functions and can be built in advance.
Operations are therefore supported in constant time. 

Moreover, each table requires at most $s\cdot u' \cdot \log s$ bits. 
Recall that we are in the sparse case defined as the case in which the size
of the universe is small relative to $n$. 
Specifically, we have that $\log (u/n) = O(\log\log n)$, hence $u = n\polylog(n)$. 
Since $B=(\log n)/(\log(u/n))$ and $n_B = (1+\delta)B$, one can show that,
under these parametrizations, $s = \polylog(n)\cdot \sqrt{n}$ and the total
number of bits each table takes is $o(n)$. 

Similarly, we can build a lookup table
that encodes the lexicographic order of the elements in each state of a $\BD$.
Each operation on the counter dictionaries is implemented by
further indexing the lookup tables with an index that denotes the position
of $c(x)$ in the $CD$.

\subsection{Elias-Fano encoding} In this section, we briefly discuss the Elias-Fano
encoding proposed in~\cite{BE20}. A bin dictionary implemented using
this encoding is referred to as a ``pocket dictionary''.
The idea is to represent each element in the universe $[u']$ as a pair $(q,r)$, where
$q\in [B]$ (the quotient) and $r\in[u'/B]$ (the remainder). 
A header encodes in unary the number of elements that have the same quotient.
The body is the concatenation of remainders. The space required is $B+n_B(1+\log (u/n))$ bits, which meets the required space bound since $B=O(n_B)$.
Similarly, a counter dictionary can be implemented by storing
the counters consecutively using an ``end-of-string'' symbol.
We use a ternary alphabet for this encoding, which requires at 
most $\Theta(B)$ bits to encode each $\CD$.

Both the $BD$s and the $CD$s fit in $O(1)$ words.
Operations in the $\BD$ and the $\CD$ require rank and select
instructions. See~\cite{BE20} for a discussion of how these operations
can be executed in constant time if the RAM model can evaluate in constant time
instructions represented as Boolean circuits with $O(\log w)$ depth and $O(w^2)$ gates.

\end{document}

%% file: includes.tex
\ifnum\fancyguy=0
\usepackage{comment}
\usepackage{ifthen}
\usepackage{xspace}
\usepackage{float}
\usepackage{subcaption}
\usepackage{amsmath}
\usepackage{amsthm}
\usepackage{amstext}
\usepackage{amssymb}
\usepackage{amsfonts}
\usepackage{bm}
\usepackage[multiple]{footmisc}
\usepackage{wrapfig}

\usepackage{url}

\usepackage[margin=1in]{geometry}
\usepackage{mathdots}
\usepackage{mathtools}
\usepackage{thmtools}
\usepackage{thm-restate}
\usepackage{tikz}
\usetikzlibrary{calc, fit, shapes,intersections,arrows,matrix, topaths}
\usetikzlibrary{decorations,decorations.pathreplacing,patterns, shadows}

\usepackage{graphicx}
\graphicspath{{./FIGS/}}
\usepackage[linesnumbered,noresetcount,vlined,ruled]{algorithm2e}

\usepackage[inline]{enumitem}
\usepackage{environ}
\usepackage{dsfont}

\usepackage{cite}
\usepackage{yfonts} 

\else

\usepackage{comment}
\usepackage{ifthen}
\usepackage{xspace}
\usepackage{float}
\usepackage{subcaption}
\usepackage{amsmath}
\usepackage{amsthm}
\usepackage{amstext}
\usepackage{amssymb}
\usepackage{amsfonts}
\usepackage{bm}
\usepackage[multiple]{footmisc}
\usepackage{wrapfig}
\usepackage[left=20mm, right=20mm, top=20mm, bottom=20mm]{geometry}
\usepackage[english]{babel}
\usepackage[utf8]{inputenc}
\usepackage{fancyhdr}
\usepackage{csquotes}

\usepackage{sectsty}
\usepackage{indentfirst}
\usepackage{url}
\allsectionsfont{\centering\mdseries\scshape}
\usepackage{mathdots}
\usepackage{mathtools}
\usepackage{thmtools}
\usepackage{thm-restate}
\usepackage{tikz}
\usetikzlibrary{calc, fit, shapes,intersections,arrows,matrix, topaths}
\usetikzlibrary{decorations,decorations.pathreplacing,patterns, shadows}

\usepackage{graphicx}
\graphicspath{{./FIGS/}}
\usepackage[linesnumbered,noresetcount,vlined,ruled]{algorithm2e}

\usepackage[inline]{enumitem}
\usepackage{environ}
\usepackage{dsfont}

\usepackage{cite}
\usepackage{yfonts} 
\usepackage{setspace}
\fi

%% file: defines.tex
\newcommand{\Err}{\text{Err}}
\newenvironment{proofsketch}{%
  \proof}{\endproof}
\renewcommand{\epsilon}{\varepsilon} 

\renewcommand{\epsilon}{\varepsilon}
\newcommand{\eps}{\varepsilon}
\newcommand{\Dict}{\textsf{Dict}}
\newcommand{\MSDict}{\textsf{MS-Dict}}

\usepackage{amsmath} \usepackage{amsthm}
\newtheorem{theorem}{Theorem}

\newtheorem{corollary}[theorem]{Corollary}

\newtheorem{definition}[theorem]{Definition}
\newtheorem{invariant}[theorem]{Invariant}
\newtheorem{claim2}[theorem]{Claim}
\newtheorem{observation}[theorem]{Observation}

\newcommand{\UU}{\mathcal{U}}
\newcommand{\maxcard}{\log ^3 n}

\newcommand{\MM}{\mathcal{M}}

\newcommand{\hq}{h^q}
\newcommand{\hr}{h^r}

\newcommand{\hb}{h^b}

\newcommand{\etal}{\textit{et al.}\xspace}
\newcommand{\size}[1]{\ensuremath{\left|#1\right|}}
\newcommand{\set}[1]{\left\{ #1 \right\}}
\newcommand{\parentheses}[1]{\left(#1\right)}
\DeclarePairedDelimiter{\floor}{\lfloor}{\rfloor}
\DeclarePairedDelimiter{\ceil}{\lceil}{\rceil}
\newcommand{\expectation}[2]{\mathbb{E}_{#1}\left[ #2 \right]}
\renewcommand{\Pr}[1]{{\mathrm{Pr}}\left[ #1 \right]}
\newcommand{\NN}{\mathbb{N}}

\DeclareMathOperator{\cquery}{\textsf{count}}

\DeclareMathOperator{\ins}{\textsf{insert}}

\DeclareMathOperator{\del}{\textsf{delete}}

\DeclareMathOperator{\op}{\textsf{op}}
\DeclareMathOperator{\out}{\textsf{out}}

\DeclareMathOperator{\BD}{\textsf{BD}}
\DeclareMathOperator{\CD}{\textsf{CD}}

\DeclareMathOperator{\polylog}{\textsf{polylog}}
\DeclareMathOperator{\poly}{\textsf{poly}}

\makeatother

\newif\ifnotes
\notestrue
\ifnotes
\newcommand{\ioana}[1]{{\ifnotes \scriptsize \textcolor{red}{Ioana: {#1}} \fi}}
\newcommand{\guy}[1]{\ifnotes {\noindent \scriptsize  \textcolor{blue} {Guy: {#1}}} \fi{}}
\else
\newcommand{\ioana}[1]{}
\newcommand{\guy}[1]{}
\fi